\documentclass[a4paper,reqno]{amsart}
\usepackage{amssymb, graphicx, enumerate, enumitem, color}
\usepackage[usenames,dvipsnames,svgnames,table]{xcolor}
\usepackage{textcomp}

\usepackage[pdftitle={On the queue-number of graphs with bounded tree-width},
            pdfauthor={Veit Wiechert},
            colorlinks=true,linkcolor=black,citecolor=black,filecolor=black,urlcolor=black
            ]
            {hyperref}

\newtheorem{theorem}{Theorem}

\newtheorem{proposition}[theorem]{Proposition}
\newtheorem*{claim}{Claim}
\newtheorem{lemma}[theorem]{Lemma}

\newcommand{\set}[1]{\{#1\}}
\newcommand{\p}{\operatorname{p}}
\newcommand{\dep}{\operatorname{d}}

\DeclareMathOperator\qn{qn}
\DeclareMathOperator\tn{tn}

\let\leq\leqslant
\let\geq\geqslant

\let\epsilon\varepsilon

\renewenvironment{enumerate}{\begin{enumorig}[label=\textup{(\roman*)}, noitemsep, topsep=2pt plus 2pt, labelindent=.2em, leftmargin=*, widest=iii]}{\end{enumorig}}

\makeatletter
\let\old@setaddresses\@setaddresses
\def\@setaddresses{\bigskip\bgroup\parindent 0pt\let\scshape\relax\old@setaddresses\egroup}
\makeatother

\begin{document}
\title{On the queue-number of graphs with bounded tree-width}

\author[V.~Wiechert]{Veit Wiechert}
\address[V.~Wiechert]{Institut f\"ur Mathematik\\
  Technische Universit\"at Berlin\\
  Berlin \\
  Germany}

\email{wiechert@math.tu-berlin.de}

\thanks{V.\ Wiechert is supported by the Deutsche Forschungsgemeinschaft within the research training group `Methods for Discrete Structures' (GRK 1408).}

\date{\today}

%

\begin{abstract}
 A \emph{queue layout} of a graph consists of a linear order on the vertices and an assignment of the edges to \emph{queues}, such that no two edges in a single queue are nested.
 The minimum number of queues needed in a queue layout of a graph is called its \emph{queue-number}.
 
 We show that for each $k\geq1$, graphs with tree-width at most $k$ have queue-number at most $2^k-1$.
 This improves upon double exponential upper bounds due to Dujmovi\'c et al.~and Giacomo et al.
 As a consequence we obtain that these graphs have track-number at most $2^{O(k^2)}$.
 
 We complement these results by a construction of $k$-trees that have queue-number at least $k+1$.
 Already in the case $k=2$ this is an improvement to existing results and solves a problem of Rengarajan and Veni Madhavan, namely, that the maximal queue-number of $2$-trees is equal to $3$.
\end{abstract}

\maketitle

\section{Introduction}
A \emph{queue layout} of a graph consists of a linear order on the vertices and an assignment of the edges to \emph{queues}, such that no two edges in a single queue are nested.
This is a dual concept to stack layouts, which are defined similarly, except that no two edges in a single stack may cross.
The minimum number of queues (stacks) needed in a queue layout (stack layout) of a graph is called its \emph{queue-number} (\emph{stack-number}).

The notion of queue-number was introduced by Heath and Rosenberg~\cite{HR92} in 1992.
Queue layouts, however, have been implicitly studied long before and have applications in fault-tolerant processing, sorting with parallel queues, matrix computations, and scheduling parallel processors (see~\cite{HLR92,HR92,Pemm92} for more details). 

In their seminal paper, Heath and Rosenberg characterize graphs admitting a $1$-queue layout as so-called \emph{arched leveled-planar} graphs and show that it is NP-hard to recognize them.
This is contrasting the situation for graphs with a $1$-stack layout, since these graphs are exactly the outerplanar graphs~\cite{BK79} and hence can be recognized in polynomial time.
Several other results relating these two types of layouts are studied in~\cite{HLR92}.
While planar graphs have stack-number at most $4$ \cite{Yann89}, it remains open whether the queue-number of planar graphs can be bounded by a constant.
This is one of the most tantalizing problems regarding queue layouts and it was conjectured to be true by Heath and colleagues~\cite{HLR92,HR92}.
In fact, they even conjecture that the queue-number can be bounded in terms of the stack-number; see~\cite{DW05} for a comprehensive study of this question.

There are some partial results towards a positive resolution of this conjecture.
Improving on an earlier result by Di Battista et al.~\cite{BFP13}, Dujmovi\'c showed that planar graphs have queue-number $O(\log n)$~\cite{Duj15}.
This result was extended to graphs with bounded Euler genus by Dujmovi\'c, Morin, and Wood~\cite{DMW-minor}.
In the more general case of graphs that exclude a fixed graph as a minor they obtained a $\log^{O(1)}n$ bound on the queue-number.

In this paper we focus on queue layouts of bounded tree-width graphs.
A comprehensive list of references to papers about further aspects of queue layouts can be found in~\cite{DW04}.

\subsection{Queue layouts and tree-width}
For several graph classes it is known that they have bounded queue-number.
For example, trees have a $1$-queue layout \cite{HR92}, outerplanar graphs a $2$-queue layout~\cite{HLR92}, partial $2$-trees (that is, series-parallel graphs) have a $3$-queue layout~\cite{RV95}, and graphs of path-width at most $p$ have a $p$-queue layout~\cite{Wood02}.

All these graphs have bounded tree-width and it was first asked by Ganley and Heath~\cite{GH01} whether there is a constant upper bound on the queue-number of bounded tree-width graphs (for the stack-number this is true as shown in~\cite{GH01}).
This question was answered in the affirmative for graphs that additionally have bounded maximum degree by Wood~\cite{Wood02}, and later in full by Dujmovi\'c and Wood~\cite{DW03} (see also \cite{DMW05}).
In the latter result, Dujmovi\'c and Wood establish the upper bound $3^k6^{(4^k-3k-1)/9}-1$ on the queue-number of graphs with tree-width at most $k$.
In fact, they provide upper bounds as solutions of a system of equations.
Giacomo et al.~\cite{giacomo} present an improved system of equations with smaller solutions for each $k\geq 1$ (without trying to find a nice expression for the corresponding upper bound), but still being double exponential in $k$.
Answering a question of Dujmovi\'c et al.~\cite{DMW05} we prove a single exponential upper bound.
\begin{theorem}\label{thm:main-upper}
 Let $k\geq 0$.
 For all graphs $G$ of tree-width at most $k$, we have that
 \[
  \qn(G)\leq 2^k-1,
  \]
 where $\qn(G)$ denotes the queue-number of $G$.
\end{theorem}
Observe that this bound is not only asymptotically much smaller than previous best bounds, it is also strong for small values of $k$.
As special cases we obtain the above mentioned results that trees and partial $2$-trees have queue-number at most $1$ and $3$, respectively. (And as we will show, $3$ is best possible in the latter case). 
Interestingly, in his PhD thesis Pemmaraju~\cite{Pemm92} supports a conjecture of him that a certain family of planar $3$-trees (the \emph{stellated triangles}) has queue-number $\Omega(\log n)$.
Of course, this conjecture has already been disproved by Dujmovi\'c et al.~with their upper bound for $k$-trees.
However, now with Theorem~\ref{thm:main-upper} we even get that planar $3$-trees (and more generally partial $3$-trees) have queue-number at most $7$.

\subsection{Track layouts}
For the proofs of their upper bounds, Dujmovi\'c et al.~and Giacomo et al.~use track layouts of graphs.
A \emph{track layout} is a partition of the vertex set into \emph{tracks} together with a linear ordering on each track, such that no two tracks induce a crossing with respect to their orderings (see Section~\ref{sec:prel} for details).
The minimum number of tracks in a track layout of a graph $G$ is called the \emph{track-number} of $G$, and denoted by $\tn(G)$.

In \cite{DMW05,DW03} the upper bound $3^k6^{(4^k-3k-1)/9}$ is actually shown for the track-number of graphs of tree-width at most $k$.
Since the authors can also show that for all graphs $G$ it holds that $\qn(G)\leq \tn(G)-1$, they obtain their bound for the queue-number from the track-number bound.
In this paper we show the following result.
\begin{theorem}\label{thm:main-tn}
 Let $k\geq 0$.
 For all graphs $G$ of tree-width at most $k$, we have that
 \[
  \tn(G)\leq (k+1)(2^{k+1}-2)^k.
 \]
\end{theorem}
Clearly, this $2^{O(k^2)}$ bound is asymptotically a big improvement upon the double exponential bound discussed before.
However, for small values of $k$ (that is, $k\in\{1,2,3\}$) better bounds are known.
For example, Giacomo et al.~\cite{giacomo} show that graphs of tree-width at most $2$ admit a $15$-track layout.

\subsection{Three-dimensional drawings}
Another reason to study queue-number and track-number is their connection to three-dimensional drawings of graphs.
A \emph{three-dimensional straight-line grid drawing} is an embedding of the vertices onto distinct points of the grid $\mathbb{Z}^3$ with edges being represented as straight lines that connect their end-vertices, such that any two straight lines intersect only if they share a common end-vertex, and a vertex can only be contained in a straight line if this vertex is the end-vertex of that line.

Using the \emph{moment curve} one can show that each graph has such a drawing.
Therefore, we would like to minimize the volume of the bounding box defined by the grid points used for the embedding.
Cohen et al.~\cite{Cohen97} showed that the complete graph $K_n$ requires $\Theta(n^3)$ volume.
Graphs with bounded chromatic number can be drawn on the three-dimensional grid with $O(n^2)$ volume, as shown by Pach et al.~\cite{Pach99}, and this is best possible for complete bipartite graphs.
The latter result was improved by Bose et al.~\cite{Bose04}, who showed that graphs with $n$ vertices and $m$ edges need at least $\frac18 (n+m)$ volume.
In particular, this implies that graphs with three-dimensional drawings of linear volume have only linear many edges.
Dujmovi\'c and Wood~\cite{DW06} showed that graphs with bounded degeneracy admit drawings with $O(n^{3/2})$ volume.
A major open problem in this area is due to Felsner et al.~\cite{FLW03} and asks whether planar graphs can be drawn with linear volume.
The best known volume bound for this problem is $O(n\log n)$ and was given by Dujmovi\'c~\cite{Duj15} (see also~\cite{DMW-minor} for an extension of this bound to apex-minor free graphs, and an $n\log^{O(1)}n$ bound for proper minor-closed families).
In~\cite{DMW05} Dujmovi\'c et al.~argue that if planar graphs have bounded queue-number, then this would imply a linear bound on the required volume for three-dimensional drawings of planar graphs.

Let us focus on graphs of bounded tree-width now.
For outerplanar graphs, which have tree-width at most $2$, Felsner et al.~\cite{FLW03} proved a linear volume bound.
Their argument is based on track layouts and a technique called ``wrapping''.
Dujmovi\'c et al.~\cite{DMW05} showed that graphs of track-number at most $t$ have a $O(t)\times O(t)\times O(n)$ drawing, implying that bounded tree-width graphs can be drawn with linear volume.
To be more precise, using the bounds on the track-number obtained by Dujmovi\'c et al.~one can deduce that graphs of tree-width at most $k$ admit $O(t_k)\times O(t_k)\times O(n)$ drawings, where $t_k=3^k\cdot 6^{(4^k-3k-1)/9}$.
The resulting volume was slightly improved by Giacomo et al.~\cite{giacomo} with their new bounds on track-number.
Now, the aforementioned volume bound in terms of the track-number combined with our Theorem~\ref{thm:main-tn} significantly reduces the required volume for bounded tree-width graphs to $2^{O(k^2)}\times 2^{O(k^2)}\times O(n)$.

\subsection{Lower bounds on the queue-number}

Not much is known with respect to lower bounds on the queue-number of graphs.
Heath and Rosenberg~\cite{HR92} give a simple argument for the fact that the queue-number is always larger than half of the edge density.
In particular, graph families with more than linear many edges have unbounded queue-number.
Gregor et al.~\cite{GSV12} show that the $n$-dimensional hypercube has queue-number at least $(\frac12-\epsilon)n-O(1/\epsilon)$, for every $\epsilon>0$.

To the author's knowledge, the best published lower bound on the queue-number of planar graphs is $2$~\cite{RV95}.
In fact, the example given is an outerplanar graph and hence has tree-width at most $2$.
In the same paper it is conjectured that there are planar graphs with queue-number $5$ (without providing deep evidence for this).

For graphs of tree-width at most $k$ the situation is similar.
It is easy to see that complete graphs and complete bipartite graphs yield examples with queue-number at least $\lfloor\frac{k+1}{2}\rfloor$, but besides that no lower bounds depending on $k$ have been discussed.
(On the other hand, this is the case for track-number~\cite{giacomo,DMW05}.)
For the special case of $2$-trees we already mentioned that their queue-number is at most $3$.
The aforementioned example also shows that there are $2$-trees with queue-number $2$.
We close this gap and thereby answer a question of Rengarajan and Veni Madhavan~\cite{RV95} with the following general lower bound.

\begin{theorem}\label{thm:lower-bound}
 For each $k\geq 2$, there is a $k$-tree with queue-number at least $k+1$.
\end{theorem}
Since $2$-trees are planar, we particularly obtain that there are planar graphs with queue-number at least $3$.

\subsection{Proof ideas and organization}
For the proof of Theorem~\ref{thm:main-upper} we make use of tree-partitions, which were introduced by Seese~\cite{Seese85} and independently by Halin~\cite{Halin91}.
A \emph{tree-partition} of a graph is a partition of its vertex set into ``bags'', combined with an underlying tree (or forest) on the bags so that each edge of the graph is either contained within a bag, or it goes along an edge of the tree.
The fact that $k$-trees admit tree-partitions such that each bag induces a $(k-1)$-tree (see~\cite{DMW05}) allows us to apply induction.
In contrast to the proofs of \cite{giacomo,DMW05}, we do not construct a track layout as an intermediate step, but directly build a queue layout of the given graph.

The rest of the paper is organized as follows.
In Section~\ref{sec:prel} we provide necessary definitions and basic propositions for our proofs.
In Section~\ref{sec:upper-bound} we prove Theorem~\ref{thm:main-upper} and \ref{thm:main-tn}.
Then we show the lower bound of Theorem~\ref{thm:lower-bound} in Section~\ref{sec:lower-bound}.
We conclude the paper with some open problems in Section~\ref{sec:problems}.

\section{Preliminaries}\label{sec:prel}
In this section we introduce the necessary definitions and basic concepts for our main result.

\subsection{Queue and Track layouts}
Let $G=(V,E)$ be a graph and let $L$ be a linear order on the vertices of $G$.
We say that edges $uv,u'v'\in E$ are \emph{nested} with respect to $L$ if $u<u'<v'<v$ or $u'<u<v<v'$ in $L$.
A set $Q$ of edges in $G$ forms a \emph{queue} with respect to $L$ if no two edges of $Q$ are nested in $L$.
A \emph{queue layout} of $G$ is a linear order $L$ on the vertices of $G$ together with a partition of the edge set of $G$ into queues with respect to $L$.
The minimum number of queues in a queue layout of $G$ is called the \emph{queue-number} of $G$, and denoted by $\qn(G)$.

There is a different access to the queue-number via $k$-rainbows.
Given a linear order $L$ on the vertices of a graph $G$, we say that the edges $a_1b_1,\ldots,a_kb_k$ form a \emph{rainbow} of size $k$ (or \emph{$k$-rainbow}) if
\[
 a_1<\cdots <a_k<b_k<\cdots <b_1
\]
in $L$.
Clearly, if $k$ is the maximum size of a rainbow in $L$, then each queue layout using $L$ as the linear order will consist of at least $k$ queues.
It is not hard to see that $k$ queues suffice in this case.
\begin{proposition}[\cite{HR92}]
 If $G$ has no rainbow of size $k+1$ with respect to a given linear order $L$, then $G$ has a queue layout using at most $k$ queues with respect to $L$.
\end{proposition}
As a consequence, the queue-number can be described as the minimum number taken over the maximal size of a rainbow in a linear order of $V(G)$.

We define \emph{track layouts} now.
Let $G$ be graph and let $\{V_i: i=1,\ldots,\ell\}$ be a partition of $V(G)$ into independent set.
A set $V_i$ combined with a linear order $<_i$ on its elements is a \emph{track} of $G$.
Then a set of tracks $\{(V_i,<_i):i=1,\ldots,\ell\}$ is called a \emph{track assignment} of $G$.
Two edges $ab$ and $cd$ form an \emph{X-crossing} in a track assignment $\{(V_i,<_i):i=1,\ldots,\ell\}$ if there are $i,j\in\{1,\ldots,\ell\}$ such that $a<_ic$ and $d<_jb$.
A track assignment without an X-crossing is called a \emph{track layout}.
The minimum number of tracks in a track layout of $G$ is the \emph{track-number} of $G$, which we denote by $\tn(G)$.

\subsection{Tree-width}
Let $G=(V,E)$ be a graph.
A \emph{tree-decomposition} of $G$ is a pair $(T,\{T_x\}_{x\in V})$ consisting of a tree $T$ and a family of non-empty subtrees of $T$, such that $V(T_x)\cap V(T_y)\neq \emptyset$ for each edge $xy\in E$.
The vertices of $T$ are called \emph{nodes}, and each node $u\in V(T)$ induces a \emph{bag} $\{x\in V: u\in T_x\}$.
The maximum size of a bag minus one is the \emph{width} of the tree-decomposition.
Then the \emph{tree-width} of $G$ can be defined as the minimum width of a tree-decomposition of $G$.

For our purposes it is convenient to follow the work of Dujmovic et al.~\cite{DMW05} and define $k$-trees as introduced by Reed~\cite{Reed03}.
Given some fixed integer $k\geq 0$, a \emph{$k$-tree} is defined recursively.
The empty graph is a $k$-tree, and each graph obtained by adding a vertex $v$ to a $k$-tree so that the adjacent vertices of $v$ form a clique of size at most $k$ is also a $k$-tree.
(Arnborg and Proskurowski~\cite{AP89} introduced $k$-trees in a slightly more restrictive way.
They start with defining a $k$-clique to be a $k$-tree, and each graph obtained from a $k$-tree by adding a vertex being adjacent to a $k$-clique is also a $k$-tree.
Sometimes the notion of \emph{strict} $k$-trees is used for this more restrictive version.)
A subgraph of a $k$-tree is called a \emph{partial $k$-tree}.
It is well-known that a graph has tree-width at most $k$ if and only if it is a partial $k$-tree.
Moreover, $k$-trees are \emph{chordal} graphs, that is, they do not contain a cycle on more than three vertices as an induced subgraph.

\subsection{Tree-partitions}
For the construction of a queue layout in our main proof, we do not use a specific tree-decomposition, but instead we use a tree-partition. 
Given a graph $G$, a \emph{tree-partition} of $G$ is a pair consisting of a tree $T$ (or forest) and a partition of $V(G)$ into sets $\{T_x: x\in V(T)\}$ being indexed by the vertices of $T$, such that for each edge $uv$ in $G$ we either have that $u,v\in T_x$ for some $x\in V(T)$, or there is an edge $xy$ of $T$ with $u\in T_x$ and $v\in T_y$.
We refer to the vertices of $T$ as \emph{nodes}, and say that $T_x$ ($x\in V(T)$) is a \emph{bag} of the tree-partition.
By $G[T_x]$ we denote the subgraph of $G$ induced by the vertices of $T_x$. 
For an example of a tree-partition see Figure~\ref{fig:tree-partition}.
A fixed tree-partition of $G$ naturally divides the edges of $G$ into two classes.
If both endpoints of an edge are contained in the same bag, then we call it an \emph{intrabag edge}.
In the other case, so if the two endpoints lie in different bags, then we call it an \emph{interbag edge}.

\begin{figure}[t]
 \centering
 \includegraphics[scale=1.2]{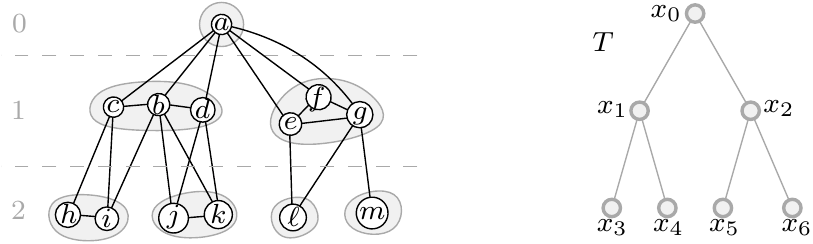}
 \caption{A tree-partition of a $3$-tree whose vertices are labeled with $a,b,\ldots,m$. The figure on the right shows the underlying tree of the tree-partition.}
 \label{fig:tree-partition}
\end{figure}

\section{Upper bounds -- Proofs of Theorem~\ref{thm:main-upper} and \ref{thm:main-tn}}\label{sec:upper-bound}
We begin this section with a proof of Theorem~\ref{thm:main-upper} and conclude it with a short proof of Theorem~\ref{thm:main-tn}.

We would like to note that it is enough to prove Theorem~\ref{thm:main-upper} for $k$-trees.
Indeed, this follows from the two facts that each graph of tree-width $k$ can be extended to a $k$-tree by adding edges to the graph (for example, by taking the \emph{chordal completion} that minimizes the size of the maximum clique), and that the queue-number of a graph does not decrease under the addition of edges.

As noted before, our queue layout construction relies on tree-partitions that capture the structure of $k$-trees.
The following theorem by Dujmovi\'c et al.~will give us such a tree-partition. 
\begin{theorem}[\cite{DMW05}]\label{thm:tree-partition}
 Let $G$ be a $k$-tree.
 Then there is a rooted tree-partition $(T,\{T_x\colon x\in V(T)\})$ of $G$ such that
 \begin{enumerate}
  \item for each node $x$ of $T$, the induced subgraph $G[T_x]$ is a connected $(k-1)$-tree,
  \item\label{item:clique} for each nonroot node $x\in T$, if $y\in T$ is the parent node of $x$ in $T$ then the vertices in $T_y$ with a neighbor in $T_x$ form a clique.
 \end{enumerate}

\end{theorem}

Let us give a brief sketch of how one can obtain a tree-partition of a connected $k$-tree $G$ as in the theorem.
Fix an arbitrary vertex $r$ of $G$ and perform a \emph{Breadth-first Search} (BFS) in $G$ starting from $r$.
For each $d\geq 0$ and each component induced by the vertices at distance $d$ from $r$, we introduce a node  and associate with this node a bag containing the vertices of the component.
Two nodes become adjacent if their corresponding sets of vertices are joined by at least one edge of $G$. 
Using the chordality of $G$ one can show that the constructed graph $T$ on the nodes is indeed a tree, and that the vertices of each bag induce a $(k-1)$-tree.
Note that the bag of the root node of $T$ contains only one vertex ($r$ in our case).
The tree-partition in Figure~\ref{fig:tree-partition} can be obtained with the described procedure by starting the BFS from vertex $a$. 

Let $(T,\{T_x\colon x\in V(T)\})$ be a rooted tree-partition as in the previous theorem.
For each nonroot node $x$ of $T$, we denote by $\p(x)$ the parent node of $x$ in $T$.
Moreover, we let $C_x$ denote the clique in $T_{\p(x)}$ according to item~\ref{item:clique} of Theorem~\ref{thm:tree-partition}.
For instance, in our example of Figure~\ref{fig:tree-partition} we have that $\p(x_3)=x_1$ and the vertices of $C_{x_3}$ are $b$ and $c$.

We are now ready to prove Theorem~\ref{thm:main-upper}.
In fact, we show the following slightly stronger result.
\begin{theorem}\label{thm:main}
 Let $k\geq 0$.
 For each $k$-tree $G$, there is a queue layout using at most $t_k=2^k-1$ queues, such that for each $v\in V(G)$, edges with $v$ as their right endpoint in the layout are assigned to pairwise different queues.
\end{theorem}
\begin{proof}
 We prove the theorem by induction on $k$.
 In the base case $k=0$, the graph $G$ has no edges and thus no queues are needed in a queue layout of $G$.
 
 Suppose now that $G$ is a $k$-tree for some $k\geq 1$, and that the theorem holds for $k-1$.
 We may assume that $G$ is connected, since we can combine layouts of different components of $G$ by putting them next to each other, and since we can resuse queues for different components. 
 Let $(T,\{T_x\colon x\in V(T)\})$ be a tree-partition of $G$ as given by Theorem~\ref{thm:tree-partition}, and denote the root of $T$ by $r$.
 Then we can assign to each node of $T$ a \emph{depth} according to its distance to $r$ (with $r$ being at depth $0$).
 We say that a vertex $v\in V(G)$ is at depth $d$ if $v$ is contained in a bag of some node at depth $d$.
 
 In the following, we first construct a linear order $L^G$ for the queue layout of $G$, and then we assign the edges to queues.
 Let us give some intuition of how we obtain $L^G$ now.
 We build $L^G$ by going through the depths one by one (starting with depth $0$).
 That is to say, given the already produced linear order of vertices at depth $d-1$, we construct a linear order of vertices at depth $d$ and append it to the right of the one already produced.
 To do so, we first specify a linear order $L_{d}^T$ on the nodes at depth $d$ in $T$, and then we replace each node $x$ in $L_d^T$ by the linear order of the layout obtained by applying induction to the $(k-1)$-tree $G[T_x]$. 
 
 Now let us be more precise.
 At depth $0$ we only have the root node $r$ of $T$, and hence we set $L_0^T$ to be the linear order consisting only of $r$. 
 We apply induction on the $(k-1)$-tree $G[T_r]$ and obtain a linear order $L_0^G$ of vertices at depth $0$ (as noted before, $T_r$ actually contains only one vertex).
 
 So suppose that we have built the linear order $L_{d-1}^G$ containing all vertices at depth at most $d-1$ in $G$.
 Let $L_{d-1}^T$ be the linear order on the nodes at depth $d-1$ that was produced in the last step of our procedure.
 We proceed by constructing $L_d^T$ now.
 
 As in a \emph{lexicographical breadth-first ordering} (Lex-BFS ordering), we order the nodes according to their parent nodes.
 That is, for nodes $x,y$ at depth $d$ we set $x<y$ in $L_d^T$ if $\p(x)<\p(y)$ in $L_{d-1}^T$.
 It remains to specify the order of nodes sharing a parent node.
 So suppose that $x_1,\ldots,x_{\ell}$ have the same parent node $y$ at depth $d-1$.
 Consider the cliques $C_{x_1},\ldots,C_{x_{\ell}}$ in $T_y$.
 For each $i\in\set{1,\ldots,\ell}$, let $c_{x_i}$ be the rightmost vertex of $C_{x_i}$ in $L_{d-1}^G$.
 Then we order $x_1,\ldots,x_{\ell}$ according to the positions of $c_{x_1},\ldots,c_{x_{\ell}}$, which means that we set $x_i<x_j$ in $L_d^T$ if $c_{x_i}$ appears before $c_{x_j}$ in $L_{d-1}^G$. 
 Nodes with the same parent node and with the same rightmost vertex in their corresponding clique are still not ordered with this rule.
 We order those nodes arbitrarily so that $L_d^T$ becomes a linear order on nodes at depth $d$.
 To illustrate this procedure, consider the following linear order, where vertices at depth at most $1$ of our example from Figure~\ref{fig:tree-partition} have been ordered so far.
 
 \begin{figure}[h]
  \centering
  \includegraphics[scale=1.2]{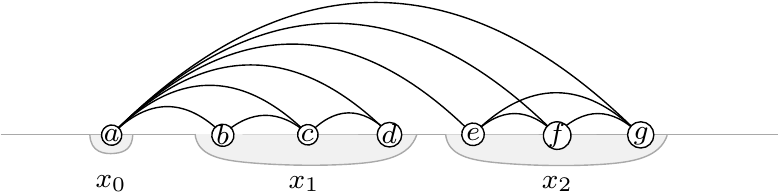}
 \end{figure}
 
 Here we have $x_1<x_2$ in $L_1^T$, and since $x_1=\p(x_3)=\p(x_4)$ and $x_2=\p(x_5)=\p(x_6)$, this implies that $x_3,x_4$ are placed before $x_5,x_6$ in $L_2^T$.
 As $c_{x_3}=c<d=c_{x_4}$ in the order, we set $x_3<x_4$ in $L_2^T$.
 The order between $x_5$ and $x_6$ in $L_2^T$ can be chosen arbitrarily as $c_{x_5}=c_{x_6}=g$.
 
 By Theorem~\ref{thm:tree-partition} we have that the bag of each node $x$ in the tree-partition induces a $(k-1)$-tree, which allows us to apply induction.
 Let $L_x$ be the linear order of the queue layout obtained in this way.
 Now we replace each node $x$ in $L_d^T$ by the linear order $L_x$. 
 We put the resulting order of vertices at depth $d$ to the right of $L_{d-1}^G$, which yields a linear order $L_d^G$ on all vertices at depth at most $d$.
 This concludes the step for vertices at depth $d$.

 Let $L^G$ be the linear order on the vertices of $G$ obtained after going through all the depths.
 Similarly, let $L^T$ be the linear order on the nodes of $T$ obtained during the procedure.
 Recall that by our applied rules, $L^T$ has the following properties. 
 For nodes $x,y\in V(T)$ with depths $\dep(x)$ and $\dep(y)$, respectively, it holds that
 \begin{align}
  \text{if }\dep(x)<\dep(y) \text{ in }T\text{, then } x<y \text{ in }L^T,\label{item:depth}\\
  \text{if }\p(x)<\p(y)\text{ in }L^T\text{, then }x<y \text{ in }L^T.\label{item:parent}
 \end{align}
 Property \eqref{item:depth} asserts that $L^T$ is a BFS ordering, and combined with property \eqref{item:parent} we have that $L^T$ is a Lex-BFS ordering.
 Therefore, no two edges of $T$ are nested in $L^T$.
 This has an immediate consequence for interbag edges as they go along edges of $T$.
 Let $uv$ and $u'v'$ be interbag edges such that $u<v$ and $u'<v'$ in $L^G$.
 Then we have the property that if $uv$ and $u'v'$ are nested in $L^G$, then $u$ and $u'$ are contained in the same bag of the tree-partition. 
 
 We need to assign the edges of $G$ to queues now.
 For convenience, let us instead first color the edges with colors from $\{1,\ldots,2t_{k-1}+1\}$ and then show that each color class induces a queue with respect to $L^G$.
 
 We start with the intrabag edges. 
 For each bag $T_x$, we color the contained edges according to the queue assignment that is given by the induction hypothesis for the $(k-1)$-tree $G[T_x]$.
 We use the colors $1,\ldots,t_{k-1}$ for this coloring (so we reuse the same colors for different bags). 
 
 Let us continue with the interbag edges now, and let $uv\in E(G)$ be one of those.
 Say, $u$ is at a smaller depth than $v$.
 Then there is a node $x$ in $T$ such that $v\in T_x$ and $u\in T_{\p(x)}$.
 If $u=c_x$, then we color $uv$ with $2t_{k-1}+1$.
 Otherwise, if $u\neq c_x$, then we color $uv$ with $i+t_{k-1}$, where $i\in\set{1,\ldots,t_{k-1}}$ is the color of the intrabag edge $uc_x$.

 \begin{claim}
  For each color $c\in\set{1,\ldots,2t_{k-1}+1}$, the edges of $G$ colored with $c$ form a queue with respect to $L^G$.
 \end{claim}
 \begin{proof}
 Suppose for a contradiction that there are edges $uv$ and $u'v'$ with color $c$ that are nested in $L^G$.
 Say, we have $u<u'<v'<v$ in $L^G$.
 
 If $c\in\set{1,\ldots,t_{k-1}}$ then $uv$ and $u'v'$ are both intrabag edges.
 However, if they lie within the same bag, then they cannot be nested as we used a valid queue layout from the induction hypothesis.
 And if they lie in different bags, then both endpoints of one edge lie before both endpoints of the other edge in $L^G$.
 Thus, the two edges are not nested in $L^G$, a contradiction.
 
 So we have $c\geq t_{k-1}+1$ and consequently $uv$ and $u'v'$ are interbag edges. 
 By the consequences of properties~\eqref{item:depth} and \eqref{item:parent} for interbag edges, it follows that $u$ and $u'$ both are contained in the same bag.
 Suppose that this is bag $T_y$, and let $x,x'\in V(T)$ be such that $v\in T_x$ and $v'\in T_{x'}$.
 Note that $u\in C_x$ and $u'\in C_{x'}$.
 We distinguish two cases now.
 
 First, suppose $c=2t_{k-1}+1$.
 Then $u$ and $u'$ are rightmost in $L^G$ among vertices of $C_x$ and $C_{x'}$, respectively.
 So we have $u=c_x$ and $u'=c_{x'}$, and hence $x\neq x'$.
 Recall that since $x$ and $x'$ share the parent $y$, they are ordered in $L^T$ according to the positions of $c_x$ and $c_{x'}$ in $L^G$.
 Thus, as $c_x=u<u'=c_{x'}$ in $L^G$, this implies $x<x'$ in $L^T$.
 It follows that vertices of $T_x$ lie before vertices of $T_{x'}$ in $L^G$, a contradiction to our assumption $v'<v$ in $L^G$.
 
 So we are left with the case $c\in\set{t_{k-1}+1,\ldots,2t_{k-1}}$.
 Let $i\in\set{1,\ldots,t_{k-1}}$ be such that $c=i+t_{k-1}$.
 This time we have $u\neq c_x$ and $u'\neq c_{x'}$.
 Since $u\in C_x$ and $u'\in C_{x'}$, it follows that $u<c_x$ and $u'<c_{x'}$ in $L^G$.
 By our coloring, edges $uc_x$ and $u'c_{x'}$ are colored with $i$.
 This implies $c_x\neq c_{x'}$ as otherwise $c_x$ is the right endpoint of two intrabag edges of the same color, which is contradicting the induction hypothesis.
 In particular, this yields $x\neq x'$.
 By our assumption that $v'<v$ in $L^G$, we conclude $x'<x$ in $L^T$.
 And since $x$ and $x'$ are ordered in $L^T$ according to the positions of $c_x$ and $c_{x'}$ in $L^G$, this in turn implies $c_{x'}<c_x$ in $L^G$.
 Together with $c_{x'}$ being the rightmost vertex of $C_{x'}$ in $L^G$, we deduce $u<u'<c_{x'}<c_x$ in $L^G$.
 It follows that the edges $uc_x$ and $u'c_{x'}$ are nested.
 However, note that both edges are contained in $T_y$ and have the same color $i$.
 This is a contradiction to the fact that we colored these edges according to the queue layout obtained by the induction hypothesis.
 This concludes the proof of the claim.
 \end{proof}
 To complete the induction step, we have to show that for each $v\in V(G)$, no two edges with $v$ as their right endpoint in $L^G$ are colored with the same color.
 Suppose for a contradiction that there are distinct edges $uv$ and $u'v$ colored with $c$ such that $u<v$ and $u'<v$ in $L^G$.
 By the induction hypothesis we cannot have $c\in\{1,\ldots,t_{k-1}\}$.
 Therefore, both edges are interbag edges and $c\in\{t_{k-1}+1,\ldots,2t_{k-1}+1\}$.
 Let $x\in V(T)$ be such that $v\in T_x$.
 Then $u$ and $u'$ are vertices of the clique $C_x$.
 Since $c_x$ is the unique vertex of $C_x$ that is connected by an edge in color $2t_{k-1}+1$ to $v$, we deduce $c\neq 2t_{k-1}+1$.
 However, then our coloring rule for the edges $uv$ and $u'v$ implies that the edges $uc_x$ and $u'c_x$ are colored with $c-t_{k-1}\in\{1,\ldots,t_{k-1}\}$.
 As $c_x$ is the rightmost vertex of $C_x$ with respect to $L^G$, we obtain that the intrabag edges $uc_x$, $u'c_x$ have the same color and the same right endpoint in $L^G$, which is a contradiction to the induction hypothesis.
 We conclude that any two edges with the same right endpoint in $L^G$ are colored with different colors.
 
 Finally, since we use $2t_{k-1}+1=2(2^{k-1}-1)+1=2^k-1$ queues in our layout of $G$, this completes the proof of the theorem. 
\end{proof}

We continue with a proof of Theorem~\ref{thm:main-tn} now.
A proper coloring of the vertices of a graph $G$ is \emph{acyclic} if any two color classes induce a forest (so each cycle receives at least three colors).
The minimum number of colors used in an acyclic coloring of $G$ is the \emph{acyclic chromatic number} of $G$.
Dujmovi\'{c} et al.~\cite{DMW05} obtained the following relationship between track-number and queue-number.
\begin{lemma}[\cite{DMW05}]\label{lem:acyclic}
 Every graph $G$ with acyclic chromatic number at most $c$ and queue-number at most $q$ has track-number
 \[
  \tn(G)\leq c(2q)^{c-1}.
 \]

\end{lemma}
It is well-known that graphs of tree-width at most $k$ have acyclic chromatic number at most $k+1$.
Using this, we immediately obtain a proof of our claimed upper bound on the track-number of bounded tree-width graphs.
\begin{proof}[Proof of Theorem~\ref{thm:main-tn}]
 Combine Theorem~\ref{thm:main-upper} and Lemma~\ref{lem:acyclic}.
\end{proof}

\section{Lower bounds -- Proof of Theorem~\ref{thm:lower-bound}}\label{sec:lower-bound}
This section is devoted to a proof of Theorem~\ref{thm:lower-bound}.
We start by introducing a two-player game between Alice and Bob on $k$-trees (where $k\geq 2$), in which Bob has to build a queue-layout of the $k$-tree to be presented by Alice.
We call it the \emph{$k$-queue game}.

The game starts with a $(k+1)$-clique and an arbitrary linear order on the vertices of this clique.
Now, each round of the game consists of two moves.
First, Alice introduces a new vertex $v$ and chooses a $k$-clique of the current graph to which $v$ becomes adjacent.
And second, Bob has to specify the position in the current layout where $v$ is inserted.
Clearly, since we start with a $(k+1)$-clique, the graphs obtained during the $k$-queue game remain $k$-trees.
It is the goal of Alice to increase the maximum size of a rainbow in the layout, while Bob tries to keep it small.
Alice \emph{wins} the $k$-queue game if Bob creates a rainbow of size $k+1$ in the layout.
We aim to show the following.
\begin{lemma}\label{lem:game-bound}
 For each $k\geq 1$, there is an integer $d_k$ such that Alice has a strategy to win the $k$-queue game within at most $d_k$ rounds.  
\end{lemma}
Before we prove this lemma, we use it to show Theorem~\ref{thm:lower-bound}.
Let us make some new definitions first.

Given a graph $H$ and a clique $C$ in $H$, we \emph{stack} on $C$ in $H$ by introducing a new vertex $v_C$ and by making $v_C$ adjacent to the vertices of $C$.
(Note that if we stack on a $k$-clique of a $k$-tree, then the resulting graph is also a $k$-tree.)
If a graph $H'$ is obtained by simultaneously stacking on each $k$-clique of $H$, then we call $H'$ the \emph{$k$-stack} of $H$.

We iteratively construct a family of $k$-trees $(G_i)_{i\in\mathbb{N}}$ now.
We let $G_0$ be a $(k+1)$-clique, and given $i\geq 1$, we define $G_i$ to be the $k$-stack of $G_{i-1}$.
Note that with this definition $G_i$ contains $G_{i-1}$ as an induced subgraph.
In fact, $G_i$ might contain several distinct induced subgraphs being isomorphic to $G_{i-1}$.
For us it is important that $G_i$ contains an \emph{intrinsic} copy $G_{i-1}'$ of $G_{i-1}$ as an induced subgraph, which is such that $G_i$ can be obtained by taking the $k$-stack of $G_{i-1}'$. 

The following lemma implies Theorem~\ref{thm:lower-bound}
\begin{lemma}
 Given $k\geq 2$, let $d_k$ be as in the statement of Lemma~\ref{lem:game-bound}.
 Then the queue-number of the $k$-tree $G_{d_k}$ is at least $k+1$.
\end{lemma}

\begin{proof}
 Consider the following variant of the $k$-queue game.
 Alice's move in a round of the variant consists of simultaneously stacking on each possible $k$-clique.
 It is then Bob's task in this round to insert all the newly introduced vertices in the current layout.
 Again, Alice wins the game when a rainbow of size $k+1$ appears in the layout.
 
 Clearly, for Bob this variant is harder than the $k$-queue game, in the sense that when Alice has a strategy to win the $k$-queue game within $d$ rounds, then she also has a strategy to win the variant within $d$ rounds.
 In particular, Lemma~\ref{lem:game-bound} also holds for the variant.
 
 Now suppose for a contradiction that there is a linear order $L$ on the vertices of $G_{d_k}$ such that there is no rainbow of size $k+1$ in $L$.
 We claim that Bob can use $L$ as an instruction to avoid rainbows of size $k+1$ during the first $d_k$ rounds in the variant of the $k$-queue game.
 
 To see this, observe that after $i$ rounds of the variant, the game graph is isomorphic to $G_i$.
 This gives rise to a strategy for Bob.
 He only has to fix induced subgraphs $H_0,H_1,\ldots,H_{d_k}$ of $G_{d_k}$ such that $H_{d_k}=G_{d_k}$ and such that $H_{i-1}$ is the intrinsic copy of $G_{i-1}$ in $H_i$ for each $i\in\{1,\ldots,d_k\}$. (Note that $H_i$ is isomorphic to $G_i$).
 Then $L|_{V(H_i)}$ is an extension of $L|_{V(H_{i-1})}$ for each $i\in\{1,\ldots,d_k\}$.
 Therefore, Bob can ensure that the linear order after $i$ rounds is equal to $L|_{V(H_i)}$.
 Indeed, he only has to read from $L$ how to extend the layout in each round.
 Applying this strategy, the linear order built after $d_k$ rounds is equal to $L$.
 As $L$ does not contain a rainbow of size $k+1$, Bob can prevent Alice from winning within the first $d_k$ rounds.
 This is a contradiction to Lemma~\ref{lem:game-bound} and completes the proof.
\end{proof}

The rest of this section is devoted to a proof of Lemma~\ref{lem:game-bound}.
We proceed with some definitions that will help us to talk about the $k$-queue game.

Let $G$ be a $k$-tree designed by Alice during the game and let $L$ be the linear order on $V(G)$ built by Bob.
Given $x,y\in V(G)$, we say that $x$ \emph{lies left of} $y$ in $L$, if $x<y$ in $L$, and we say that $x$ \emph{lies right of} $y$ in $L$, otherwise.
We denote the leftmost and the rightmost vertex of a subgraph $H$ of $G$ with respect to $L$ by $\ell(H)$ and $r(H)$, respectively.
An edge $e$ of $G$ \emph{covers} a subgraph $H$ of $G$ in $L$ if $\ell(e)\leq \ell(H)<r(H)\leq r(e)$ in $L$.
The edge $e$ \emph{strictly covers} $H$ if we have $\ell(e)<\ell(H)$ and $r(H)<r(e)$ in $L$.
Suppose Alice chooses to stack on the clique $C$ in her next move.
Then we say that Bob \emph{goes inside} $C$ if he places the new vertex $v_C$ such that $\ell(C)<v_C<r(C)$ in the layout.
Otherwise, we say that Bob \emph{goes outside} $C$.
If Bob places $v_C$ such that $r(C)<v_C$ in the layout, then he \emph{goes to the right outside} of $C$.

We continue by developing a strategy for Alice to win the $k$-queue game within a finite number of rounds.
Whenever we write that Alice \emph{can force} Bob to make certain moves, then we mean that she has a strategy to win the game unless Bob does these moves.

\begin{lemma}\label{lemma:stack-outside}
 For any $k$-clique $C$ in the game graph and any positive number $d$, Alice can force Bob to go outside some $k$-clique $C^*$, which is covered by the edge $\ell(C)r(C)$, for at least $d$ times.
\end{lemma}
\begin{proof}
 We describe a strategy for Alice to enforce the claimed behavior of Bob.
 First, Alice starts to stack on the clique $C$ in her moves.
 If Bob does not go inside $C$ for $d$ rounds, then $C$ fulfills the desired requirements.
 
 So suppose that Bob goes inside $C$ with the vertex $v_C$ so that $\ell(C)<v_C<r(C)$ in $L$.
 Note that the vertices in the set $V(C)\setminus\{\ell(C)\}\cup\{v_C\}$ form a $k$-clique $C'$.
 For the next rounds, Alice keeps on stacking on $C'$.
 Again, if Bob does not go inside $C'$ for $d$ rounds, then we are done with clique $C'$.
 So suppose that he goes inside $C'$ with the vertex $v_{C'}$.
 Then the vertices in $V(C')\setminus\{r(C)\}\cup \{v_{C'}\}$ form a $k$-clique $C''$ that is strictly covered by the edge $\ell(C)r(C)$.
 
 Now observe that if Alice applies the above strategy to $C''$ instead of $C$, and Bob keeps on avoiding to go outside $k$-cliques as before, then we will see $k$-clique being strictly covered by the edge $\ell(C'')r(C'')$ after several rounds.
 Clearly, if Alice is repeating this strategy, then we will see a rainbow of size $k+1$ in the layout unless Bob goes outside some $k$-clique being covered by $\ell(C)r(C)$ for at least $d$ times, as claimed.
\end{proof}

\begin{lemma}\label{lem:outside-right}
 Let $C$ be a $k$-clique in the game graph with vertices $v_1,\ldots,v_k$ such that $v_1<\cdots<v_k$ in the layout.
 Assume that Alice can force Bob to go to the right outside of $C$ at least $2k+1$ many times.
 Then Alice can enforce the existence of a vertex $v_{k+1}$ in the layout such that
 \begin{enumerate}
  \item\label{item:order} $v_k<v_{k+1}$ in the layout,
  \item\label{item:clique-adj} $v_{k+1}$ is adjacent to $C$, and 
  \item\label{item:outside-right} Alice can force Bob to go to the right outside of $C'$ arbitrary many times, where $C'$ denotes the $k$-clique on the vertices $v_1,v_3,\ldots,v_{k+1}$. 
 \end{enumerate}
\end{lemma}
\begin{proof}
 By assumption, Alice can force Bob to place $2k+1$ vertices $p_1,\ldots,p_k,v_{k+1},q_1,\ldots,q_k$, which are adjacent to $C$, to the right of $C$ in the layout.
 Let us suppose that
 \[
  p_k<\cdots<p_1<v_{k+1}<q_k<\cdots<q_1
 \]
 in the layout.
 For each $i\in\{1,\ldots,k\}$, we let $e_i:=p_iv_i$ and $e_i':=q_iv_i$.
 Observe that the edges $e_1,\ldots,e_k$ and $e_1',\ldots,e_k'$ form rainbows of size $k$.
 
 We claim that $v_{k+1}$ fulfills the requirements of the statement.
 Clearly, $v_{k+1}$ lies to the right of $v_k$ in the layout and it is adjacent to $C$, so \ref{item:order} and \ref{item:clique-adj} hold.
 Denote the $k$-clique on the vertices $v_1,v_3,\ldots,v_{k+1}$ by $C'$.
 Even stronger than condition~\ref{item:outside-right}, we can show that whenever Alice introduces a vertex $v_{C'}$ being adjacent to $C'$, then Bob loses unless he puts $v_{C'}$ to the right of $C'$ (that is, to the right of $v_{k+1}=r(C')$).
 
 So suppose that Bob places $v_{C'}$ to the left of $v_{k+1}$.
 Then let $j\in\{1,\ldots,k+1\}$ be minimal such that $v_{C'}<v_j$ in the layout (see Figure~\ref{fig:forcing-left} illustrating an example with $k=4$ and $j=3$).
 We obtain that the edges $e_1',\ldots,e_{j-1}',v_{C'}v_{k+1},e_j,\ldots,e_k$ form a rainbow of size $k+1$ in the layout (in Figure~\ref{fig:forcing-left} this is the rainbow consisting of red edges), implying that Bob lost the game.
 Therefore, condition~\ref{item:outside-right} holds.
\end{proof}

\begin{figure}[t]
 \centering
 \includegraphics{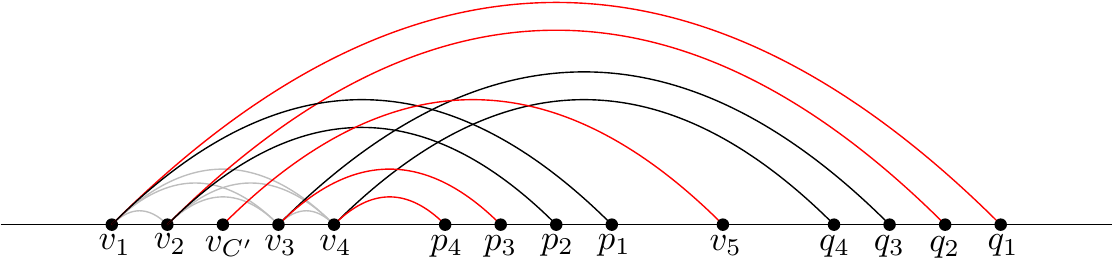}
 \caption{Situation in the $4$-queue game (not all existing edges are depicted). If Bob places $v_{C'}$ to the left of $v_5$, then this creates a $5$-rainbow.}
 \label{fig:forcing-left}
\end{figure}

Later, we will show that Alice can reach a winning configuration in the $k$-queue game by using the previous lemma.
This configuration is described in the following lemma.

\begin{lemma}\label{lemma:winning-conf}
 Suppose that there are edges $e,e',e''$, and a $k$-clique $C$ in the game graph such that
 \[
  \ell(e)\leq \ell(e')<r(e')<\ell(C)<r(C)<\ell(e'')<r(e'')\leq r(e)
 \]
 in the layout built by Bob (see Figure~\ref{fig:winning-conf} for an illustration of such a situation).
 Then Alice has a strategy to win the current $k$-queue game within a finite number of rounds.
\end{lemma}
\begin{proof}
 Given the configuration of the statement, we describe a strategy for Alice to win the game.
 Alice starts by applying the strategy of Lemma~\ref{lemma:stack-outside} to enforce a $k$-clique $C'$ being covered by the edge $\ell(C)r(C)$, such that Bob is forced to go outside $C'$.
 Note that $C'$ and the edges $e,e',e''$ also build a configuration as described in the statement of the lemma.
 So we may assume that $C$ is already the clique on which Bob is forced to go outside.
 In the following, let $v_1,\ldots,v_k$ be the vertices of $C$ such that $v_1<\cdots<v_k$ in the layout.
 
 Next, Alice keeps on stacking on $C$ until there are $2k-1$ vertices adjacent to $C$ that all lie to the left of $C$, or that all lie to the right of $C$ (as Bob has to go outside $C$, this happens after at most $4k-3$ rounds).
 By symmetry, we may assume that these $2k-1$ vertices lie left of $\ell(C)$.
 Using the pigeonhole principle we obtain that either there are $k$ such vertices lying to the left of $\ell(e)$, or $k$ such vertices lying between $\ell(e)$ and $\ell(C)$.
 
 Let us consider the first case now.
 So we have $k$ vertices $p_1,\ldots,p_k$ adjacent to $C$ such that
 \[
  p_k<\cdots<p_1<\ell(e)\leq \ell(e')<r(e')<v_1<\cdots<v_k
 \]
 in the layout.
 Then the edges $e',p_1v_1,\ldots,p_kv_k$ form a rainbow of size $k+1$, and hence Alice wins the game.
  
 In the second case Bob has placed $k$ vertices $p_1,\ldots,p_k$ being adjacent to $C$ such that
 \[
  \ell(e)<p_k<\cdots<p_1<v_1<\cdots<v_k<r(e)
 \]
 in the layout.
 However, in this case the edges $p_1v_1,\ldots,p_kv_k,e$ form a rainbow of size $k+1$.
 
 This shows that Alice has a winning strategy once the configuration in the statement of the lemma occurs during the game. 
\end{proof}

\begin{figure}[t]
 \centering
 \includegraphics{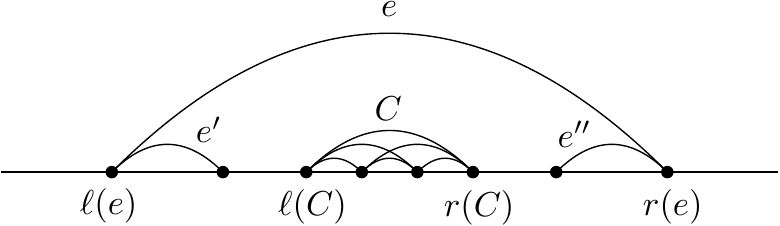}
 \caption{Winning configuration for Alice.}
 \label{fig:winning-conf}
\end{figure}

We are now ready to combine the previous lemmas to give a proof of Lemma~\ref{lem:game-bound}.
\begin{proof}[Proof of Lemma~\ref{lem:game-bound}]
 We describe a strategy for Alice to win the $k$-queue game.
 Using a $k$-clique of the initial graph in the game and the strategy of Lemma~\ref{lemma:stack-outside}, Alice can enforce a $k$-clique $C_1$ on which Bob has to go outside arbitrary many times.
 Next, Alice keeps on stacking on $C_1$ until Bob has placed $2k+1$ of the newly introduced vertices either to the left of $C_1$, or to the right of $C_1$.
 By symmetry, we may assume that the latter occurs.
 
 Observe that $C_1$ fulfills the assumptions of Lemma~\ref{lem:outside-right}.
 Starting with $C_1$, we now describe how Alice can iteratively apply the strategy of this lemma.
 Let $v_1,\ldots,v_k$ be the vertices of $C_1$ such that $v_1<\cdots<v_k$ in the layout.
 Then by Lemma~\ref{lem:outside-right} Alice can enforce a vertex $v_{k+1}$ to the right of $v_k$, such that $v_{k+1}$ is adjacent to $C_1$ and Bob is forced to go to the right outside of the $k$-clique $C_2$ consisting of the vertices $v_1,v_3,\ldots,v_{k+1}$.
 
 Clearly, Alice can now apply the strategy of Lemma~\ref{lem:outside-right} to $C_2$.
 So suppose that Alice goes on like this for another three times starting with $C_2$, and denote the three newly enforced vertices by $v_{k+2}, v_{k+3}$, and $v_{k+4}$.
 Then we have $v_1<\cdots<v_{k+4}$ in the layout, and with their introduction the new vertices became adjacent to the following vertices: vertex $v_{k+2}$ to $v_1,v_3,\ldots,v_{k+1}$, vertex $v_{k+3}$ to $v_1,v_4,\ldots,v_{k+2}$, and vertex $v_{k+4}$ to $v_1,v_5,\ldots,v_{k+3}$.
 Figure~\ref{fig:final-conf} shows this situation for $k=4$.
 
 Next we show that the resulting layout contains the winning configuration of Lemma~\ref{lemma:winning-conf}.
 To see this, let $e:=v_1v_{k+4}$, $e':=v_1v_2$, and $e'':=v_{k+3}v_{k+4}$.
 Now note that $e,e',e''$ and the $k$-clique formed by the vertices $v_3,\ldots,v_{k+2}$ build such a winning configuration.
 
 Therefore, Alice can apply the strategy of Lemma~\ref{lemma:winning-conf} and wins the $k$-queue game.
 By the arguments used for the proofs of Lemmas~\ref{lemma:stack-outside}-\ref{lemma:winning-conf}, it is also clear that Alice can exploit her winning strategy within a number of rounds that only depends on $k$.
 This completes the proof.
\end{proof}

\begin{figure}[t]
 \centering
 \includegraphics{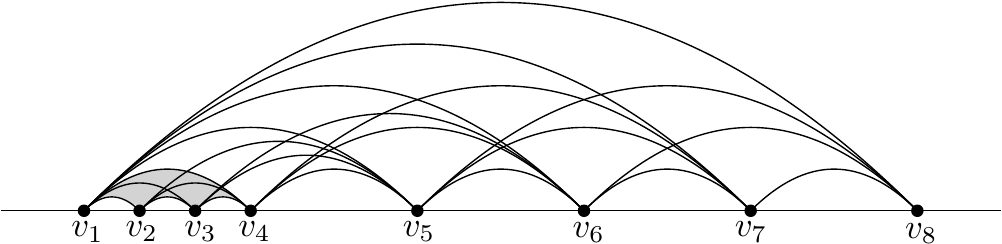}
 \caption{Situation in the $4$-queue game after applying Lemma~\ref{lem:outside-right} four times (starting with the $4$-clique on $v_1,v_2,v_3,v_4$). It contains the winning configuration of Figure~\ref{fig:winning-conf} that is formed by the edges $v_1v_8,v_1v_2,v_7v_8$ and the clique on $v_3,v_4,v_5,v_6$.}
 \label{fig:final-conf}
\end{figure}

\section{Open Problems}\label{sec:problems}
In this paper we showed a single exponential upper bound on the queue-number of graphs with tree-width at most $k$.
It remains open whether this bound can be reduced to a bound that is polynomial in $k$.
Regarding our theorem on the lower bound, it seems unlikely that $k+1$ is the right answer for the maximal queue-number of $k$-trees.
A quadratic lower bound would already be an exciting improvement.

As mentioned in the introduction, it remains open whether planar graphs have bounded queue-number.
The current best upper bound of $\mathcal{O}(\log n)$ is due to Dujmovic~\cite{Duj15}.
From below we showed the existence of planar graphs with queue-number at least $3$.
This a surprising large gap for such a popular class of graphs.

Concerning the track-number, the analogue upper bound problems are unsolved as well.

\section{Acknowledgment}
We are very grateful to Piotr Micek for the many discussions on the topic.

\bibliographystyle{plain}
\bibliography{queuenumber}

\end{document}